\documentclass[12pt]{article}
\usepackage{amsthm,amsmath,amssymb,amsfonts,delarray}
\usepackage{times}
\usepackage{graphicx}
\usepackage{color}
\usepackage{multirow}
\usepackage{appendix}
\usepackage{rotating}
\usepackage{bbm}
\usepackage{latexsym}
\newcommand{\captionfonts}{\normalsize}

\makeatletter  
\long\def\@makecaption#1#2{%
  \vskip\abovecaptionskip
  \sbox\@tempboxa{{\captionfonts #1: #2}}%
  \ifdim \wd\@tempboxa >\hsize
    {\captionfonts #1: #2\par}
  \else
    \hbox to\hsize{\hfil\box\@tempboxa\hfil}%
  \fi
  \vskip\belowcaptionskip}
\makeatother

\newtheorem{theorem}{Theorem}[section]

\theoremstyle{definition}

\theoremstyle{remark}

\def\DIC{\textsc{dic}}

\def\WAIC{\textsc{waic}}
\def\PCIC{\textsc{pcic}}

\def\MCMC{\textsc{mcmc}}
\def\loocv{\textsc{loocv}}
\def\IPW{\textsc{ipw}}

\newcommand{\R}{\mathbb{R}}
\newcommand{\Cov}{\mathrm{Cov}}
\newcommand{\Ep}{\mathbb{E}}
\newcommand{\Eppos}{\Ep_{\mathrm{pos}}}

\newcommand{\Covpos}{\mathrm{Cov}_{\mathrm{pos}}}

\newcommand{\Vpos}{\mathbb{V}_{\mathrm{pos}}}

\renewcommand{\Pr}{\mathbb{P}}
\renewcommand{\tilde}{\widetilde}
\renewcommand{\hat}{\widehat}

\begin{document}

%%\bibliographystyle{natbib}

%%\def\spacingset#1{\renewcommand{\baselinestretch}%
%%{#1}\small\normalsize} \spacingset{1}

%%%%%%%%%%%%%%%%%%%%%%%%%%%%%%%%%%%%%%%%%%%%%%%%%%%%%%%%%%%%%%%%%%%%%%%%%%%%%%

%\hspace{13.9cm}1
\hspace{13.9cm}

\ \vspace{20mm}\\

{\noindent \LARGE Posterior Covariance Information Criterion \protect\\ for Weighted Inference}

\ \\
{\bf \large Yukito Iba  and Keisuke Yano}\\
{The Institute of Statistical Mathematics.}\\
%

%\ \\[-2mm]
\\
Markov chain Monte Carlo;
predictive model selection;
quasi-Bayesian prediction

\newpage

\thispagestyle{empty}
\markboth{}{NC instructions}
\ \vspace{-0mm}\\
%
%Abstract
\begin{center} {\bf Abstract} \end{center}
For predictive evaluation based on quasi-posterior distributions,
we develop a new information criterion, the posterior covariance information criterion ({\PCIC}). {\PCIC} generalises the widely applicable information criterion ({\WAIC}) so as to effectively handle predictive scenarios where likelihoods for the estimation and the evaluation of the model may be different. 
A typical example of such scenarios
is the weighted likelihood inference, including
prediction under covariate shift and counterfactual prediction.
%, and prediction with surrogate score function.
%The first two scenarios employ weighted likelihood and the last one utilises a surrogate score function instead of the full likelihood.
The proposed criterion utilises a posterior covariance form and is computed by using only one Markov chain Monte Carlo run.
Through numerical examples, we demonstrate how {\PCIC} can apply in practice.
Further, we show that {\PCIC} is asymptotically unbiased to the quasi-Bayesian generalization error under mild conditions
in weighted inference  with both regular and singular statistical models.%by elucidating the relationship between {\PCIC} and leave-one-out cross validation. 
%%%%%%%%%%%

\section{Introduction}
In statistical research, predictive model selection is a central topic. 
Since Akaike's information criterion \cite{Akaike_1973} initiated this field of study, various information criteria have been suggested to evaluate the quality of out-of-sample prediction (e.g., \cite{Takeuchi_1976, Konishi_Genshiro_1996, Efron_2004}). For this purpose, many methods of cross-validatory assessments (e.g., \cite{Gelfand_Dey_1994,Vehtari_etal_2017}) have also been proposed.    

The evaluation of Bayesian predictive models is a topic of developing interest in this field \cite{BDA3}. Because the Markov chain Monte Carlo ({\MCMC}) has become a popular tool for Bayesian inference, evaluating model performance using posterior samples generated by {\MCMC} is quite convenient. Pioneering proposals in this direction include the deviance information criterion ({\DIC}; \cite{Spiegelhalter_etal_2002}), Bayesian leave-one-out cross validation ({Bayesian-\loocv}; \cite{Gelfand_Dey_1994,Vehtari_etal_2017}), and widely applicable information criterion ({\WAIC}; \cite{Watanabe_2010,Watanabe_book_mtbs,Millar}). These criteria allow the predictive distribution to be evaluated by using samples from a single run of posterior simulation with primary data; no additional simulations with ``leave-one-out” data are required. 

However, these criteria 
%suggested for Bayesian models 
presume that training and test data are sampled under the same condition. Such an assumption is not acceptable in several critical predictive scenarios. A typical example is prediction under \textit{covariate shift} \cite{Shimodaira_2000, Sugiyama_Krauledat_Muller_2007, Yamazaki_etal_2007}, where the distributions of covariates in regression are changed during the training and test stages. It is well known that using weighted likelihood optimized for test samples produces better prediction results (see \cite{Shimodaira_2000}, and see also \S\ref{sec:applications} of this paper). Another case of interest is {\it counterfactual prediction}; it is also called {\it causal inference} in traditional terminology. Causal inference and predictive model evaluation have long been developed as almost independent fields. However, there is increasing interest in predictive model evaluation under counterfactual situations (e.g., \cite{Platt_etal_2013,Baba_Kanamori_Ninomiya_2017}). Weighted likelihood adjusting counterfactual circumstances is often employed in the training phase of counterfactual prediction (e.g., \cite{Robins_etal_2000}).

This study aims to provide a computationally efficient method for evaluating the prediction with Bayesian models in these settings. To achieve this, we extend {\WAIC} to handle weighted inference, where likelihoods for the estimation and the evaluation of the model could differ. In that case, the penalty term of {\WAIC} is generalised to a posterior covariance form,
which characterises the consequent criterion given in \S\ref{sec:PCIC} and will be referred to as {\it Posterior Covariance Information Criterion} ({\PCIC}).

The suggested criterion, {\PCIC}, share advantages of {\WAIC} in that it is calculated only by using posterior samples, and the result is numerically stable with respect to influential observations that make the variance of the importance sampling techniques large (c.f.,~\cite{Peruggia_1997}). It can be calculated by a single run of posterior simulation based on the original data. So, the proposed method has significant advantages over the conventional cross-validatory evaluation \cite{Sugiyama_Krauledat_Muller_2007}, where training should be performed repeatedly for each training set. Further, the use of model-by-model analytical computations of information quantities occurring in the bias-correction term \cite{Shimodaira_2000,Baba_Kanamori_Ninomiya_2017} is eliminated, which can enhance domain users' ability to assess their own models when utilising information criteria.
%could detract from domain users' ability to assess their own models when utilising information criteria, is eliminated.

Another important feature of {\WAIC} inherited by {\PCIC} is that it can deal with {\it singular statistical models} in the sense of \cite{Watanabe_book_mtbs}, where the posterior distributions are not well approximated by a normal distribution. 
We show that {\PCIC} is an asymptotically unbiased estimator of the quasi-Bayesian generalization error 
under mild conditions in 
weighted inferences with both regular and singular statistical models; see
Theorem \ref{thm: consistency of PCIC in regular models}.
We provide conditions for the theorem in Appendix \ref{appendix: regularity conditions} and the proof of in Appendix \ref{proofoftheorem}.

\section{Posterior Covariance Information Criterion}\label{sec:PCIC}

In this section, we present {\PCIC} and its basic features in a generic form. In the subsequent section, we explain how to use it in a variety of predictive scenarios.

First, to represent situations in that likelihoods for the estimation and the evaluation of the model are different, we utilise a quasi-Bayesian framework to define the predictive distribution; the working posterior in such a case is called the quasi-posterior distribution \cite{Chernozhukov_Hong_2003}, 
or the generalised posterior distribution \cite{Bissiri_etal_2016}.

Suppose that we have $n$ independent and identically distributed observations $Y=(Y_{1},\ldots,Y_{n})$ from a sample space $\mathcal{Y}$ and weights $\{w_{i}>0: i=1,\ldots,n\}$ for observations.
On the basis of these observations,
we define a quasi-posterior density on 
$d$-dimensional parameter space $\Theta$ in $\R^{d}$ associated with arbitrary observation-wise score functions $s_{i}(\cdot,\cdot):\mathcal{Y}\times \Theta\to\mathbb{R}$ 
and a prior density $\pi(\theta)$ as:
\begin{align*}
    \pi(\theta\,;\, Y)
    =\frac{\exp\{ \sum_{i=1}^{n}s_{i}(Y_{i},\theta)\}\pi(\theta) }
    {\int \exp\{\sum_{i=1}^{n}s_{i}(Y_{i},\theta')\}\pi(\theta')d\theta'}.
\end{align*}
%Examples of observation-wise score functions are illustrated in \S \ref{sec:applications}.
Using this, we define the quasi-Bayesian predictive density as:
\begin{align*}
    h_{i,\pi}(\cdot \mid Y) := \int h_{i}(\cdot \mid \theta) \pi(\theta \,;\,Y) d\theta , \quad i=1,\ldots,n,
\end{align*}
where $h_{i}(\cdot \mid \theta)$ is an observation-wise probability density on $\mathcal{Y}$ parameterised by $\theta\in \Theta$.
Our aim is to estimate the weighted quasi-Bayesian generalisation error
\begin{align*}
    G_{n}:=\Ep_{\tilde{Y}}\left[-
    \frac{1}{n}\sum_{i=1}^{n}w_{i}\log h_{i,\pi}(\tilde{Y}_{i}\mid Y)\right],
\end{align*}
where $\tilde{Y}=(\tilde{Y}_{1},\ldots,\tilde{Y}_{n})$ is an independent copy of $Y$ 
and $\Ep_{\tilde{Y}}$ is the expectation with respect to $\tilde{Y}$.

\textbf{Example.}
Let us illustrate an example of the weighted quasi-Bayesian generalization error.
Assume we have the current i.i.d.~pairs $\mathcal{D}:=\{(Y_{i},X_{i}):i=1,\ldots,n\}$ of response variables and covariates with $X_{i}$ following from $p_{\mathrm{train}}(x)$,
and 
we want to evaluate the generalisation error of a predictive density $f(\cdot\,;\, \tilde{X}_{n+1} ,\mathcal{D})$ 
with $\tilde{X}_{n+1}$ following from $p_{\mathrm{test}}(x)$.
Then, by the importance sampling formula, 
we obtain the weighted generalisation error as follows:
\begin{align*}
&\Ep_{\tilde{X}_{n+1}\sim p_{\mathrm{test}}}
\Ep_{\tilde{Y}_{n+1}\mid \tilde{X}_{n+1}}
[
-\log f(\tilde{Y}_{n+1}\,;\,\tilde{X}_{n+1},\mathcal{D})
]\\
&=
\frac{1}{n}\sum_{i=1}^{n} 
\Ep_{X_{i}\sim p_{\mathrm{train}}} \frac{p_{\mathrm{test}}(X_{i})}{p_{\mathrm{train}}(X_{i})}
\Ep_{\tilde{Y}_{n+1}\mid X_{i}}
[
-\log f(\tilde{Y}_{n+1}\,;\,X_{i},\mathcal{D})
]\\
&=\Ep_{\tilde{Y}}\left[-
    \frac{1}{n}\sum_{i=1}^{n}w_{i}
    \log f(\tilde{Y}_{i}\,;\,X_{i},\mathcal{D})\right],
\end{align*}
where $w_{i}:=p_{\mathrm{test}}(X_{i})/p_{\mathrm{train}}(X_{i})$.
This situation is known as prediction under covariate shift and will be discussed in \S\ref{sec:applications},
where several choices of predictive density $f$ different from Bayesian predictive density are also discussed.

To estimate the weighted quasi-Bayesian generalization error, 
we present the \textit{posterior covariance information criterion} ({\PCIC}) defined by
\begin{align}
    \PCIC &= -\sum_{i=1}^{n}
    \frac{w_i}{n}\log \Eppos[h_{i}(Y_{i}\,\mid\, \theta)] \,+\,
    \sum_{i=1}^{n} 
    \frac{w_i}{n}\,\Covpos 
    \left[
    \log h_{i}(Y_{i}\,\mid\,\theta)\,,\, s_{i}(Y_{i}\,,\,\theta)
    \right],
\label{eq_pcic}
\end{align}
where let operations $\Eppos$ and $\Covpos$ denote
expectation and covariance with respect to the quasi-posterior distribution, respectively.
An important feature of {\PCIC} is that the training phase does not necessarily employ the full likelihood as the validation phase, which enables us to treat a wide range of predictive scenarios, as shown in \S\ref{sec:applications}.

The following theorem shows that
{\PCIC} is an asymptotically unbiased estimator of the quasi-Bayesian generalisation error $G_{n}$.
Our theorem admits $\log h_{i}$ and $s_{i}$ to be singular in the sense of \cite{Watanabe_book_mtbs}; i.e., information matrices $\Ep_{Y}[\nabla_{\theta}\log h_{i}(Y_{i}\mid\theta)\nabla^{\top}_{\theta}\log h_{i}(Y_{i}\mid\theta)]$,
$\Ep_{Y}[\nabla_{\theta}s_{i}(Y_{i}\mid\theta)\nabla^{\top}_{\theta}s_{i}(Y_{i}\mid\theta)]$ are singular; for details, see conditions in Appendix \ref{appendix: regularity conditions}.
The proof is given in Appendix \ref{proofoftheorem}.

\begin{theorem}\label{thm: consistency of PCIC in regular models}
Under conditions in Appendix \ref{appendix: regularity conditions},
we have
\begin{align*}
\Ep_{Y}[G_{n}]-\Ep_{Y}[\PCIC]=o(n^{-1}),
\end{align*}
where $\Ep_{Y}$ is the expectation with respect to $Y$.
\end{theorem}

Of course, the conditions for the theorem do not allow $h_{i}$ and $s_{i}$ to be arbitrary. But,  
these are satisfied in the weighted  inferences with both regular and singular statistical models discussed in the subsequent section.

\textit{Remark.}
{\PCIC} is a natural generalisation of {\WAIC}
\begin{align*}
    \WAIC := -\frac{1}{n}\sum_{i=1}^{n}\log \Eppos[h(Y_{i}\mid\theta)]
    +\frac{1}{n}\sum_{i=1}^{n}\Vpos[\log h(Y_{i}\mid\theta)],
\end{align*}
where 
$\log h(\cdot\mid\cdot)$ is a log likelihood,
and
$\Vpos$ is the variance with respect to the posterior distribution.
This is verified by setting all score functions to the log likelihood $s_{1}(\cdot\,,\,\cdot)=s_{2}(\cdot\,,\,\cdot)=\cdots=s_{n}(\cdot\,,\,\cdot)=\log h(\cdot\,\mid\,\cdot)$ and setting $w_{1}=\cdots=w_{n}=1$.

\textit{Remark.}
The proof follows the standard machinery for the singular statistical models developed by \cite{Watanabe_2010,Watanabe_2010_b,Watanabe_book,Watanabe_book_mtbs}.
Yet, 
an important and non-trivial difference is to employ 
the Stein identity in the presence of correlation \cite{Stein_1981,Cochrane_2001}:
Suppose that $(N_{1},N_{2})$ follows a bivariate Gaussian distribution and $F$ is a differentiable function with $\Ep[|F'(N_{1})|]$. Then, we have
\begin{align}
\Cov[F(N_{1}),N_{2}]=\Cov[N_{1},N_{2}]\Ep[F'(N_{1})].
\label{Stein identity}
\end{align}
This identity partially explains why the posterior covariance form appears in {\PCIC};
for details, see Appendix \ref{proofoftheorem}.

\section{Applications}
\label{sec:applications}

This section explain two uses of {\PCIC}.

\subsection{Covariate shift adaptation}

First, we will look at a {\PCIC} application in the prediction under covariate shift \cite{Shimodaira_2000}.
The prediction under covariate shift has received a lot of attention in many fields including bioinformatics, spam filtering, brain-computer interfacing, and econometrics.

Let $\{(Y_{i},X_{i}):i=1,\ldots,n\}$ be i.i.d.~pairs of a response variable and a covariate, and we assume a parametric model $\{h(Y_{n+1}\mid X_{n+1}\,,\,\theta): \theta\in\Theta\}$ for a response variable $Y_{n+1}$ given a covariate $X_{n+1}$. We introduce a pair of distributions $p_{\mathrm{train}}$ and $p_{\mathrm{test}}$ that express the distributions of the covariate $X_{i}$s in the training and test phases, respectively. For simplicity's sake, we assume that the ratio $r(x):=p_{\mathrm{test}}(x)/p_{\mathrm{train}}(x)$ is known. 
Then, the quasi-posterior distribution associated with $r(X_{i})$s is given by
\begin{align*}
    \pi_{\lambda}(\theta \,;\,\{Y_{i},X_{i}\})\propto \exp
    \left\{ \sum_{i=1}^{n}
    r^{\lambda}(X_{i}) \,\log h(Y_{i}\mid X_{i}\,,\,\theta)\right\}\pi(\theta),
\end{align*}
where $\pi(\theta)$ denotes the prior density. Here, according to the previous studies \cite{Shimodaira_2000,Sugiyama_Krauledat_Muller_2007}, we employed the tilting parameter $\lambda\ge 0$ to control the trade-off between consistency and stability; if $\lambda=1$, 
the working quasi-likelihood in the training phase is unbiased to the full likelihood in the test phase, but its variance can be large (e.g., \cite{Shimodaira_2000}).
Meanwhile, setting $\lambda=0$ corresponds to the Bayesian inference using the unweighted likelihood, which usually leads to inconsistent but stable results.
{\PCIC} in this setup is derived from the generic expression \eqref{eq_pcic}  
by setting $w(X_{i})=r(X_{i})$, $s_{i}(Y_{i},X_{i},\theta)=r^{\lambda}(X_{i}) \,\log h(Y_{i}\mid X_{i}\,,\,\theta)$, which is given as 
\begin{align*}
    \PCIC = 
    -\sum_{i=1}^{n}
    \frac{r(X_{i})}{n}
    \log \Eppos[ h(Y_{i}\mid X_{i}\,,\,\theta)] \,+\,
    \sum_{i=1}^{n}
    \frac{r^{1+\lambda}(X_{i})}{n}
    \,\Vpos  \left[
    \log h(Y_{i}\mid X_{i}\,,\,\theta)\right].
\end{align*}
We can select the value of the tilting parameter $\lambda$ by minimising {\PCIC}.

We illustrate the proposed method by selecting the tilting parameter $\lambda$ in a regression problem described in \cite{Sugiyama_Krauledat_Muller_2007}. 
We set the true distribution of observations as follows:
\begin{align*}
    Y_{i}=\mathrm{sinc}(X_{i}) + \varepsilon_{i}, \quad i=1,\ldots,n,n+1,\ldots,n+m,
\end{align*}
where $\varepsilon_{i}$s follows $N(0,\sigma^{2})$ with $\sigma^{2}=(0.25)^2$. The first $n=50$ covariates $\{X_{i}:i=1,\ldots,n\}$ from $N(0,1)$ and the latter $m=50$ covariates $\{X_{i+n}:i=1,\ldots,m\}$ from $N(0.5,(0.3)^{2})$ are used as the training and test data, respectively.

On the basis of the above setting of the true distribution of observations,
we conduct two numerical experiments with two different models.
Let $\phi(x;\mu,\tau^{2})$ be the normal density with mean $\mu$ and variance $\tau^{2}$.

\textbf{Experiment 1.}
The first experiment discusses
the values of {\PCIC} and {\WAIC} using the regular linear regression model 
\[\{\phi(Y;f_{\theta}(X),\sigma^{2}) \,\text{with}\, f_{\theta}(X)=\theta_{1}+\theta_{2}X :\theta_{1},\theta_{2}\in\mathbb{R}\}.\]
Here we use the prior density give by $\pi(\theta)=\phi(\theta_{1};0,1)\phi(\theta_{2};0,1)$.

\begin{figure}[h]
    \centering
    \includegraphics[width=130mm]{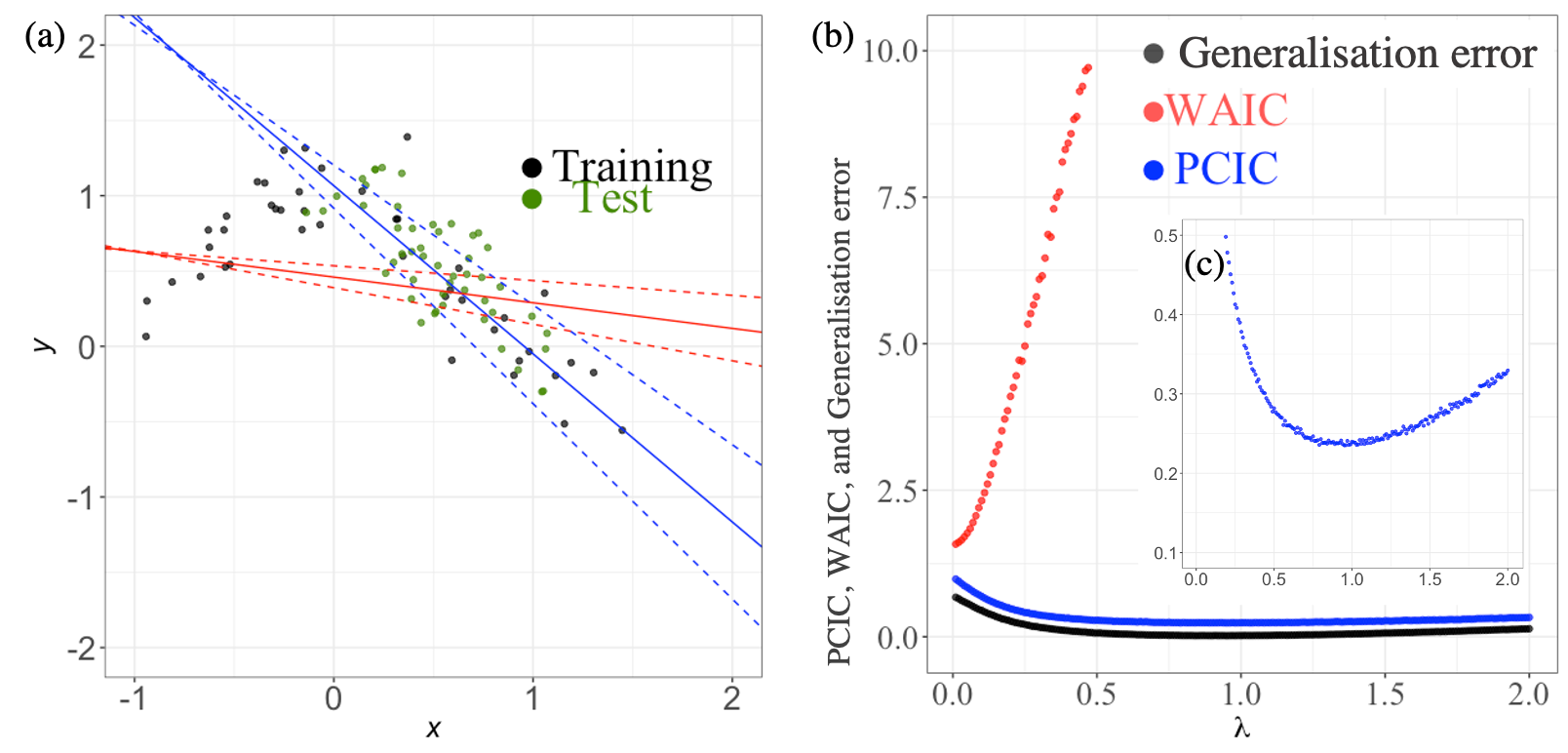}
    \caption{
    An example of the use of {\PCIC} in covariate shift adaptation with regular statistical models: 
    (a) The posterior mean of $f_{\theta}(x)$ and the $95\%$ credible interval for {\PCIC}-best $\lambda$ are shown by a blue solid line and two dashed lines, respectively, while those for {\WAIC}-best $\lambda$ are shown by red solid and dashed lines, respectively. The training and test data are shown as black and green circles, respectively.
    (b) {\PCIC}, {\WAIC}, and the generalisation error for $\lambda\in\{0.01 i\,:\,i=1,\ldots,200\}$. (c) Enlarged plot of {\PCIC} curve.
    \label{fig:covariateshift_A}
    }
\end{figure}

Figure \ref{fig:covariateshift_A} shows the qualitative comparison of {\PCIC} and {\WAIC} with regular statistical models; here {\WAIC} uses the unweighted likelihood, and both are defined using quasi-posterior samples,
where we obtained the quasi-posterior distributions exactly by using the conjugacy.
Figure \ref{fig:covariateshift_A} (a) displays means and their uncertainties of the {\PCIC} and {\WAIC}-best quasi-posterior distribution with the sets of training and test data.
Figure \ref{fig:covariateshift_A} (b) depicts the values of {\PCIC}s, {\WAIC}s, and the quasi-Bayesian generalisation error
$(1/m)\sum_{i=n+1}^{n+m}\left\{-\log h_{i,\pi_{\lambda}}(Y_{i} \mid X_{i})\right\}$.
Figure \ref{fig:covariateshift_A} (c) presents an enlarged plot of {\PCIC} curve which emphasizes the bias-variance trade-off in choosing $\lambda$.
{\PCIC} is close to the quasi-Bayesian generalisation error, whereas {\WAIC}
 cannot accommodate the change in $\lambda$ (Figure \ref{fig:covariateshift_A} (b)).
These features are reflected in the resultant predictions shown in Figure \ref{fig:covariateshift_A} (a).

\textbf{Experiment 2.}
The second experiment discusses
the selection of $\lambda$ by {\PCIC} and {\WAIC} using the singular regression model
\[\{\phi(Y;f_{\theta}(X),\sigma^{2}) \,\text{with}\, f_{\theta}(X)=\theta_{1}\,\text{reLU}(\theta_{2}+\theta_{3}X)+\theta_{4}:\theta_{i}\in\mathbb{R},i=1,2,3,4\},\]
where $\text{reLU}(x):=\max\{0,x\}$.
We use the prior density given by
$\pi(\theta)=\prod_{i=1}^{4}\phi(\theta_{i};0,1)$.

\begin{figure}[h]
    \centering
    \includegraphics[width=130mm]{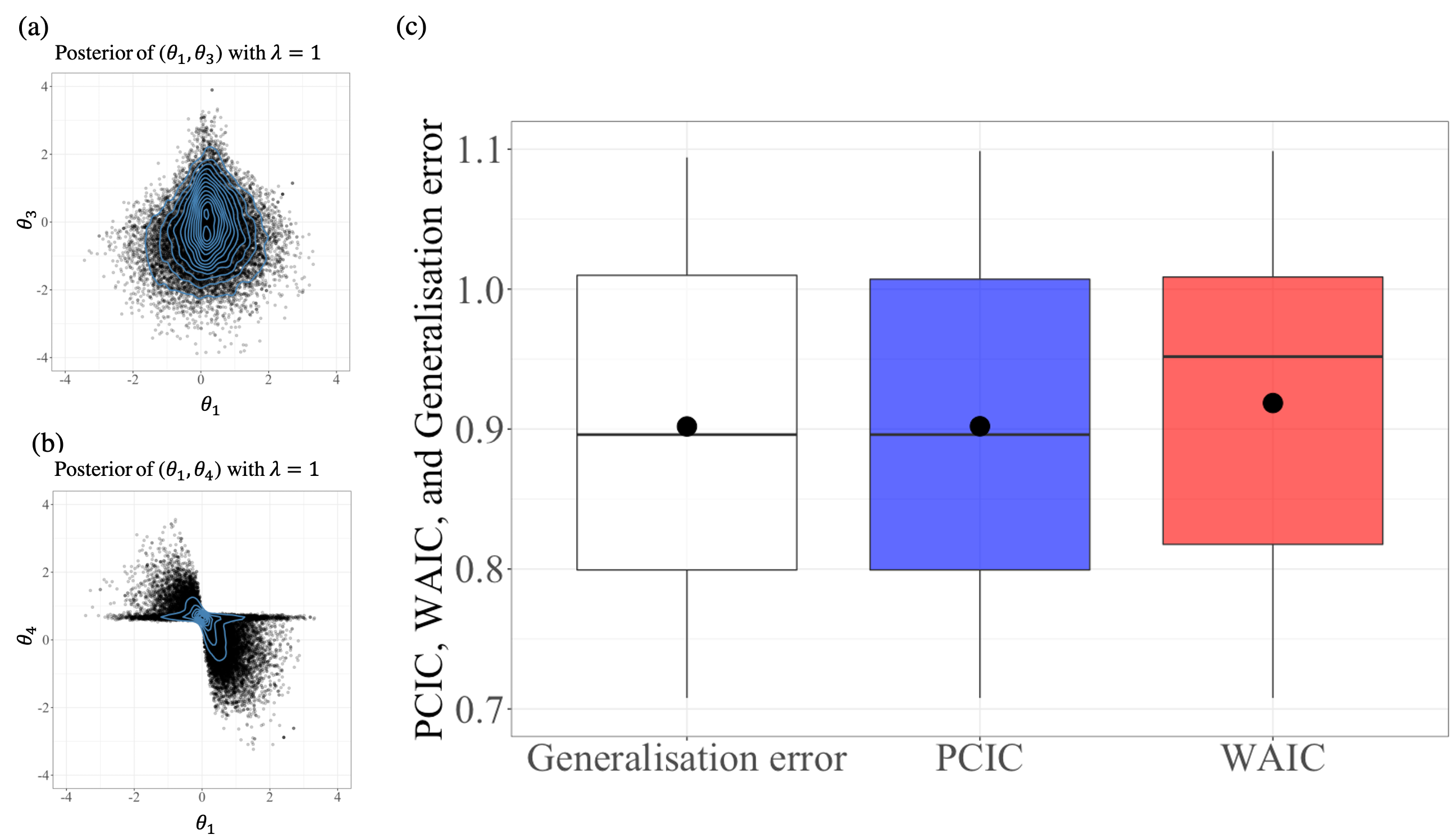}
    \caption{
    An example of the use of {\PCIC} in covariate shift adaptation with singular statistical models: 
    (a) The quasi-posterior of $(\theta_1,\theta_3)$ with $\lambda=1$, where the blue curves represent the contours. 
    (b) The quasi-posterior of $(\theta_1,\theta_4)$ with $\lambda=1$, where the blue curves represent the contours. 
    (c) Boxplots (with means represented by black dots) of the generalization errors
    with the selected $\lambda\in\{0,0.5,1.0\}$ based on {\PCIC}, {\WAIC}, and the expected generalisation error.
    \label{fig:covariateshift_B}
    }
\end{figure}

Figure \ref{fig:covariateshift_B} shows the quantitative comparison of {\PCIC} and {\WAIC} with singular statistical models.
Figures \ref{fig:covariateshift_B} (a) and (b) display the quasi-posterior distributions obtained by the Metropolis--Hastings algorithm,
where the initial value of each parameter is zero, 
the proposal distribution for each parameter is the Gaussian with mean zero and variance $(0.25)^2$,
the number of iterations is $1000000$,
the number of the burnin is $10000$,
and 
the number of the thinning is $50$.
These figures exhibit the non-normality of the quasi-posterior distribution in this set-up.
Figure \ref{fig:covariateshift_B} (c) depicts the boxplots of
the quasi-Bayesian generalisation errors
$(1/m)\sum_{i=n+1}^{n+m}\left\{-\log h_{i,\pi_{\hat{\lambda}}}(Y_{i} \mid X_{i})\right\}$
with $\hat{\lambda}$ selected among $\{0.0,0.5,1.0\}$ on the basis of 
{\PCIC}s, {\WAIC}s, and the expected quasi-Bayesian generalisation error. This figure shows that the means (as well as the medians) of the minimum quasi-Bayesian generalisation errors and the quasi-Bayesian generalisation errors with $\hat{\lambda}$ selected by {\PCIC}
are very close compared to those with $\hat{\lambda}$ selected by {\WAIC}, which suggests 
that, for prediction under the covariance shift,
 {\PCIC} works better than {\WAIC} even in singular statistical models.

\subsection{Causal inference}

Next, we discuss model evaluation in counterfactual situations. Several criteria for selecting the marginal structural model \cite{Robins_etal_2000} have been developed, for example, cross-validatory assessment \cite{Brookhart_varderLaan_2006}, quasi-likelihood information criterion \cite{Platt_etal_2013}, and a $C_{p}$ criterion \cite{Baba_Kanamori_Ninomiya_2017}.

Here, we develop a quasi-Bayesian version of the $C_{p}$ criterion in \cite{Baba_Kanamori_Ninomiya_2017} by {\PCIC}, and compare {\PCIC} to the $C_{p}$ criterion called $\mathrm{w}C_{p}$ proposed by \cite{Baba_Kanamori_Ninomiya_2017}.
We focus on an inverse probability weighted ({\IPW}) estimation using the propensity score \cite{Rosenbaum_Rubin_1983}, and utilise a Bayesian framework.

We use the following notation: $H$ and $n$ denote the number of treatments and sample size, respectively. The observed data consist of a set of quartets  $\{(Y_{i},X_{i},T_{i},Z_{i})\,:\,i=1,\ldots,n\}$, where
$i$ denotes an individual, $Y_{i}$ represents an observed outcome, $X_{i}=(X_{i}^{(1)},\ldots,X_{i}^{(H)})$ are the set of baseline covariates, $T_{i}=(T_{i}^{(1)},\ldots,T_{i}^{(H)})\in\{0,1\}^{H}$ are treatment assignments ($T_{i}^{(h)}=1$ means treatment $h$ is applied and $0$ otherwise), and $Z_{i}$ 
is a confounder vector.

To depict a counterfactual scenario, we employ a potential outcome $Y^{(h)}=\{Y_{i}^{(h)}:i=1,\ldots,n\}$ with the treatment $h$ being applied; an observed outcome $Y_{i}$ equals to $Y_{i}^{(h)}$ when the treatment $h$ is applied to the individual $i$. Then, $h(\cdot\mid X,Z,\theta)$ is assumed to be a conditional parametric density of a potential outcome given a covariate and a confounder. Assuming that the propensity score $e_{i}^{(h)}:=\Pr(T_{i}^{(h)}=1\mid Z_{i})$ is known, the quasi-posterior distribution is defined as
\begin{align*}
    &\pi(\theta \,;\,\{(Y_{i},X_{i},T_{i},Z_{i}):i=1,\ldots,n\})\\
    &\propto \exp
    \left\{  \sum_{i=1}^{n}\sum_{h=1}^{H}\frac{T_{i}^{(h)}}{e_{i}^{(h)}}\log h(Y_{i}^{(h)}\mid X_{i}^{(h)}\,,\,Z_{i}\,,\,\theta) \right\}\pi(\theta),
\end{align*}
where $\pi(\theta)$ is a prior density.

Here we adapt the prediction in the original population
and consider 
a weighted loss 
\begin{align*}
\textsc{wloss}&:=
\Ep_{\tilde{Y}}\left[
- \sum_{i=1}^{n}\sum_{h=1}^{H}\frac{T_{i}^{(h)}}{e_{i}^{(h)}}\log h_{\pi}(\tilde{Y}_{i}^{(h)}\mid Y_{1},\ldots,Y_{n}\,,\,X_{i}^{(h)}\,,\,Z_{i})
\right]
\end{align*}
with $\tilde{Y}_{i}$ being an independent copy of $Y_{i}$. Expectation of \textsc{wloss} 
corresponds to a Bayesian predictive version of \textsc{wrisk} in \cite{Baba_Kanamori_Ninomiya_2017}.
\begin{figure}[h]
    \centering
    \includegraphics[width=130mm]{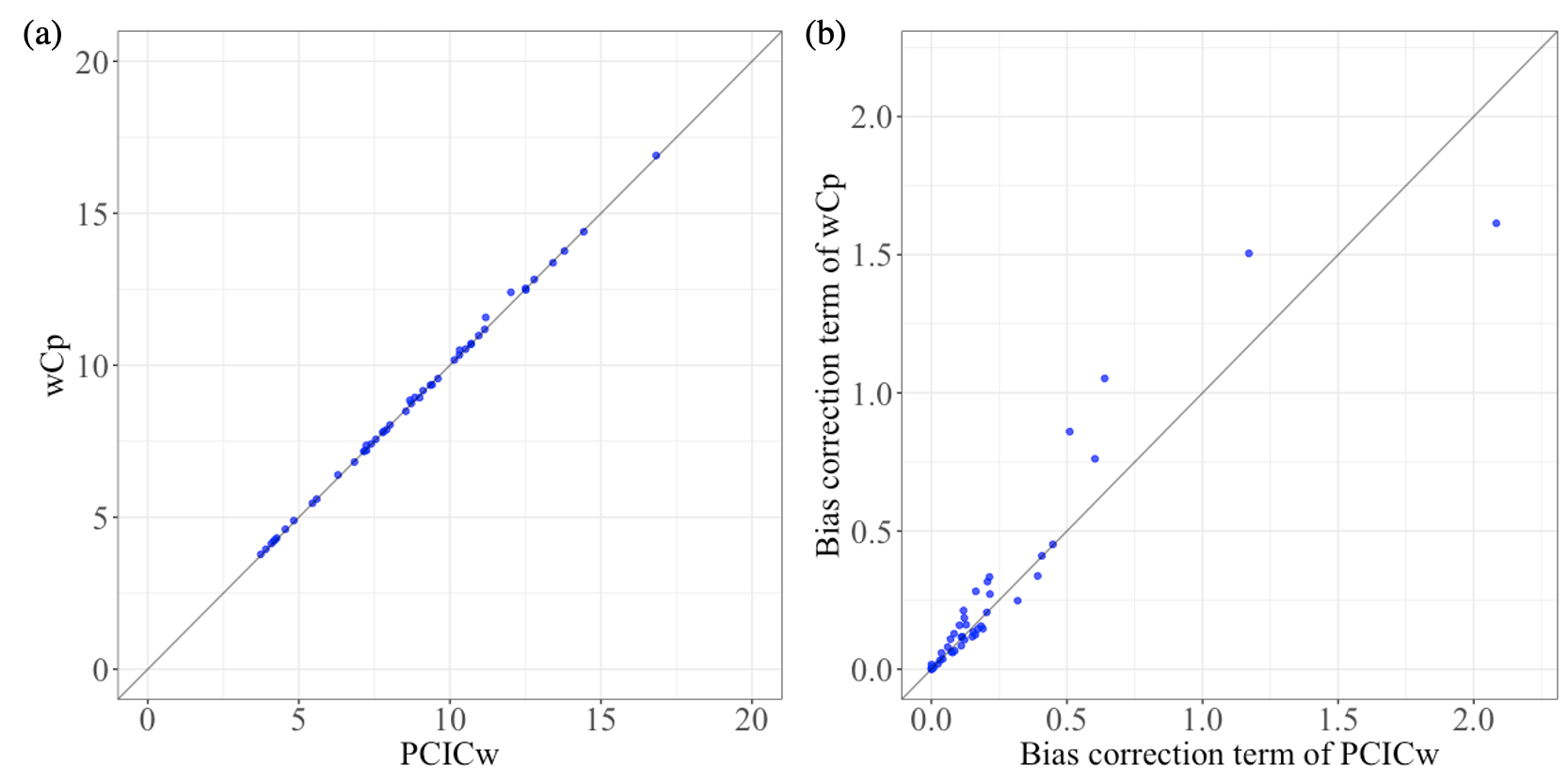}
    \caption{Application of {\PCIC} to the causal inference: (a) the correspondence between ${\PCIC}_{\mathrm{w}}$ and $\mathrm{w}C_{p}$. (b) Comparison of the bias correction terms.}
    \label{fig:semiparametriccausal}
\end{figure}
Now, we can apply the generic theory in \S\ref{sec:PCIC}, which offers an information criterion using the weighted loss for quasi-Bayesian prediction: 
\begin{align*}
    {\PCIC}_{\mathrm{w}} =& 
    -\frac{1}{n}\sum_{i=1}^{n}\sum_{h=1}^{H}\frac{T_{i}^{(h)}}{e_{i}^{(h)}}
    \log \Eppos[ h(Y^{(h)}_{i}\mid X^{(h)},Z_{i},\theta)] 
    \\
    &\,+\,
    \frac{1}{n}\sum_{i=1}^{n}
    \sum_{h=1}^{H}
    \left(\frac{T_{i}^{(h)}}{e_{i}^{(h)}}\right)^{2}\,\Vpos  \left[
    \log h(Y^{(h)}_{i}\mid X^{(h)}_{i}\,,\,Z_{i}\,,\,\theta)\right].
\end{align*}

We compare ${\PCIC}_{\mathrm{w}}$ to $\mathrm{w}C_{p}$ of the {\IPW} estimators given in \cite{Baba_Kanamori_Ninomiya_2017}, where $\mathrm{w}C_{p}$  is asymptotically unbiased to \textsc{wrisk} of {\IPW} estimators for linear regression models with variance known.
Figure \ref{fig:semiparametriccausal} compares ${\PCIC}_{\mathrm{w}}$ and  $\mathrm{w}C_p$
when we employ the linear regression setup described in \cite{Baba_Kanamori_Ninomiya_2017}. We create samples as
\begin{align*}
    Y^{(h)}_{i} = 1 + x^{(h)} + 0.5 (x^{(h)})^{2} + z + \epsilon, \, i=1,\ldots,n=50,\,h=1,\ldots,H=6,
\end{align*}
where $z$ and $\epsilon$ are independently distributed according to
$\mathrm{Un}(-3^{1/2},3^{1/2})$ and $N(0,1)$, respectively.
Here we set the density of $N(\theta_{1}+\theta_{2}x^{(h)}+\theta_{3}(x^{(h)})^{2},2)$ to $h(Y\mid x^{(h)},Z,\theta)$, and
set $N(0,1000\cdot I_{3\times 3})$ to the prior distribution, where $I_{3\times 3}$ is the $3\times3$ identity matrix. 
The values of ${\PCIC}_{\mathrm{w}}$ and $\mathrm{w}C_{p}$ (and their bias-correction terms) are very similar as shown in Figure \ref{fig:semiparametriccausal}.

\section{Conclusion}
\label{sec:conclusion}

We have proposed {\PCIC}, a new information criterion that generalises {\WAIC} to the weighted inference such as covariate shift adaptation and causal inference. 
{\PCIC} shares the favorable features of {\WAIC}: it can be computed with a single run of {\MCMC} and is asymptotically unbiased to the Bayesian generalisation error even in singular statistical models.
{\PCIC} can also be used for predictions with surrogate score functions.

The proposed criteria show that information necessary for predictive evaluation of Bayesian methods is represented as a posterior covariance. Our research demonstrates that this type of representation can be ubiquitously used in a variety of predictive scenarios.

\section*{Acknowledgement}

The authors would like to thank anonymous referees for their comments.
The authors would like to thank Yoshiyuki Ninomiya and Yusaku Ohkubo for fruitful discussions. The authors would also like to thank Shintaro Hashimoto and Tetsuya Takabatake for their helpful comments on the early version of the manuscript. We used an \textsf{R} package bayesQR \cite{Benoit_etal_2017}.
This work was partly supported by JSPS KAKENHI 19K20222, 21K12067, JST CREST JPMJCR1763, and MEXT JPJ010217.

\bibliographystyle{plain}
\bibliography{main}

 \section*{Appendix}

\appendixpageoff
\appendixtitleoff
\renewcommand{\appendixtocname}{Supplementary material}
\begin{appendices}

\section{Conditions for the theorem}
\label{appendix: regularity conditions}

In this appendix,
we state the conditions for the theorem and then provide discussions on them.
To handle singular statistical models with which the posterior distributions are not well approximated by a normal distribution,
we use several concepts (such as the standard forms and the relatively finite variances) used in \cite{Watanabe_book_mtbs}.

\textbf{Conditions}:
We here state conditions for the theorem.

We first prepare several notations for conditions. 
Let 
\begin{align*}
\Theta_{0}&:=\left\{
\theta_{0}\in\Theta\,:\,
\Ep_{Y}\left[\sum_{i=1}^{n}s_{i}(Y_{i},\theta)\right]\text{ is maximized at }\theta_{0}
\right\}\\
\Theta_{0}'&:=\left\{
\theta_{0}\in\Theta\,:\,
\Ep_{Y}\left[\sum_{i=1}^{n}w_{i}\log h_{i}(Y_{i}\mid \theta)\right]\text{ is maximized at }\theta_{0}
\right\}
\end{align*}
Fix arbitrary $\theta_{0}\in\Theta_{0}$ and,
for $i=1,\ldots,n$,
let
\begin{align*}
f_{s,i}(y,\theta)
=s_{i}(y,\theta_{0})-s_{i}(y,\theta)
\,\,\text{and}\,\,
f_{h,i}(y,\theta)=
w_{i}\{
\log h_{i}(y\mid \theta_{0})
-
\log h_{i}(y\mid \theta)
\}.
\end{align*}
For $i=1,\ldots,n$, let
\begin{align*}
S_{i}(\theta)=\Ep_{Y}[f_{s,i}(Y_{i},\theta)]
\,\,\text{and}\,\,
H_{i}(\theta)=\Ep_{Y}[f_{h,i}(Y_{i},\theta)],
\end{align*}
Let
\begin{align*}
\bar{S}(\theta)=\frac{1}{n}\sum_{i=1}^{n}S_{i}(\theta)
\,\,\text{and}\,\,
\bar{H}(\theta)=\frac{1}{n}\sum_{i=1}^{n}H_{i}(\theta),
\end{align*}
and let 
\begin{align*}
\bar{S}_{n}(\theta; Y)
=\frac{1}{n}\sum_{i=1}^{n}f_{s,i}(Y_{i},\theta)
\,\,\text{and}\,\,
\bar{H}_{n}(\theta; Y)
=\frac{1}{n}\sum_{i=1}^{n}f_{h,i}(Y_{i},\theta).
\end{align*}

We then make the following conditions.
\begin{enumerate}
    \item[ (C1) ] $\Theta_{0}=\Theta_{0}'$.
    \item[ (C2) ] There exist both a compact subset $\mathcal{M}$ of an analytic manifold and a proper analytic function $g$ from $\mathcal{M}$ to $\Theta$ such that in each local coordinate of $\mathcal{M}$ and for each $i=1,\ldots,n$, we have
    \begin{align*}
    S_{i}(g(u))&= A_{s,i}(u)u^{2k},
    \\
    \pi(g(u))|g'(u)|&=b(u)|u^{h}|,
    \\
    H_{i}(g(u))&=A_{h,i}(u)u^{2k},
    \end{align*}
    where $A_{s,i}(u)$, $b(u)>0$, and $A_{h,i}(u)>0$ are positive analytic functions,
    and 
    $d$-dimensional multi-indices $k=(k_{1},\ldots,k_{d})$ and $h=(h_{1},\ldots,h_{d})$ depend on local coordinates. Here $u^{2k}$ indicates $u^{2k}:=\prod_{i=1}^{d} u_{i}^{2k_{i}}$.
    \item[ (C3) ] For each $i=1,\ldots,d$, 
    $f_{s,i}(y,\theta)$ and $f_{h,i}(y,\theta)$ have relatively finite variances:
    there exists a positive constant $\kappa$ such that for an arbitrary $\theta\in \Theta$,
    \begin{align*}
    \Ep_{Y}[\{f_{s,i}(Y_{i},\theta)\}^{2}]
    \le \kappa S_{i}(\theta)
    \,\,\text{and}\,\,
    \Ep_{Y}[\{f_{h,i}(Y_{i},\theta)\}^{2}]
    \le \kappa H_{i}(\theta).
    \end{align*}
    \item[ (C4) ] For a sufficiently large $l$, we have
    \begin{align*}
&\sup_{n}\max_{i=1,\ldots,n}\left(
\Ep_{Y}
\sup_{\theta \in \Theta}|f_{s,i}(Y_{i},\theta)|^{l}
\right)^{1/l}<\infty \,\text{and}\\
&\sup_{n}\max_{i=1,\ldots,n}\left(
\Ep_{Y}
\sup_{\theta\in\Theta}|f_{h,i}(Y_{i},\theta)|^{l}
\right)^{1/l}<\infty
    \end{align*}
\item[ (C5) ] We have $0<\inf_{n}\min_{i=1,\ldots,n}w_{i}\le \sup_{n}\max_{i=1,\ldots,n}w_{i}<\infty$.
\end{enumerate}

We further make the condition related to weighted empirical processes.
By Conditions 2 and 3,
for each $i=1,\ldots,n$,
there exist functions $a_{s,i}(y,u)$
and $a_{h,i}(y,u)$ that are analytic with respect to $u$ and 
\begin{align}
f_{s,i}(y,g(u))=a_{s,i}(y,u)u^{k}
\,\,\text{and}\,\,
f_{h,i}(y,g(u))=a_{h,i}(y,u)u^{k}.
\label{eq: forms of f}
\end{align}
Let 
\[\bar{A}_{s}(u):=\frac{1}{n}\sum_{i=1}^{n}A_{s,i}(u)
\,\,\text{and}\,\,
\tilde{a}_{s,i}(y,u):=a_{s,i}(y,u)/\bar{A}_{s}(u).
\]
Let
\[
\bar{A}_{h}(u):=\frac{1}{n}\sum_{i=1}^{n}A_{h,i}(u)
\,\,\text{and}\,\,
\tilde{a}_{h,i}(y,u):=a_{h,i}(y,u)/
\bar{A}_{h}(u).
\]
Let
\[
\xi_{s,n}(u):=\frac{1}{\sqrt{n}}\sum_{i=1}^{n}
\{u^{k}-\tilde{a}_{s,i}(Y_{i},u)\}
\quad\text{and}\quad
\xi_{h,n}(u):=\frac{1}{\sqrt{n}}\sum_{i=1}^{n}
\{u^{k}-\tilde{a}_{h,i}(Y_{i},u)\}
\]
Using these, we make the following additional conditions.
\begin{enumerate}
\item[ (C6) ] The following uniform law of large numbers holds:
\begin{align*}
\Ep_{Y}\sup_{\theta,\theta'\in\Theta}
\left|\frac{1}{n}\sum_{i=1}^{n}f_{s,i}(Y_{i},\theta)f_{h,i}(Y_{i},\theta')-\frac{1}{n}\sum_{i=1}^{n}\Ep_{Y}[f_{s,i}(Y_{i},\theta)f_{h,i}(Y_{i},\theta')]\right|
=o(1)
\end{align*}
\item[ (C7) ] The following weak convergence holds:
\begin{align*}
\begin{pmatrix}
\xi_{s,n}(u)\\
\xi_{h,n}(v)
\end{pmatrix}
\Rightarrow
\begin{pmatrix}
\xi_{s}(u)\\
\xi_{h}(v)
\end{pmatrix}
\end{align*}
    in $\ell^{\infty}(\mathcal{M} \times \mathcal{M}\to\mathbb{R}^{2})$,
where 
$\ell^{\infty}(\mathcal{M}\times \mathcal{M}\to \mathbb{R}^{2})$ is the set of all uniformly bounded functions $z:\mathcal{M}\times \mathcal{M}\to \mathbb{R}^{2}$,
and
$(\xi_{s},\xi_{h})^{\top}$ is the $\mathbb{R}^{2}$-valued Gaussian random field in $\ell^{\infty}(\mathcal{M}\times\mathcal{M}\to\mathbb{R}^{2}).$ 
\item[(C8)] For a sufficiently large $l$, we have
    \begin{align*}
&\sup_{n}\max_{i=1,\ldots,n}\left(
\Ep
\sup_{u \in \mathcal{M}}|a_{s,i}(Y_{i},u)|^{l}
\right)^{1/l}<\infty \,\text{and}\\
&\sup_{n}\max_{i=1,\ldots,n}\left(
\Ep
\sup_{u\in\mathcal{M}}|a_{h,i}(Y_{i},u)|^{l}
\right)^{1/l}<\infty
\end{align*}
\end{enumerate}

\textbf{Discussions on conditions}:
We here discuss conditions we use.
\begin{enumerate}
\item[(C1)]: Condition (C1) ensures that the optimal parameter sets for the training and the test phases are identical. In the weighted predictive inference with weights not depending on $\theta$, this condition is satisfied.
\item[(C2)]: 
By Hironaka's resolution of singularities, non-zero analytic functions $F(\theta)$ with $\{\theta:F(\theta)=0\}$ non-empty has the standard form, that is, 
there exist both a compact subset $\mathcal{M}_{F}$ of an analytic manifold and a proper analytic function $g_{F}$ from $\mathcal{M}_{F}$ to $\Theta$ such that in each local coordinate of $\mathcal{M}_{F}$, we have
\begin{align*}
F(g(u))=A_{F}(u)u^{2k_{F}}
\quad\text{and}\quad
|g'(u)|=B_{F}(u)|u^{2h_{F}}|,
\end{align*}
where $A_{F},B_{F}>0$, and $k_{F}$ and $h_{F}$ are $d$-dimensional multi-indices.
Condition (C2) ensures that the standard forms after Hironaka's resolution of singularities for $\Ep_{Y} \log h_{i}(Y_{i}\mid \theta)$ and $\Ep_{Y}[s_{i}(Y_{i},\theta)]$ have the same multi-index. This condition seems somewhat restrictive. Yet, the weighted inference with both singular and regular statistical models
satisfies this condition. So, this is a mild condition in the weighted inference.
\item[(C3)]: Condition (C3) ensures the relatively finite variances of all quantities $f_{s,i}$ and $f_{h,i}$. This ensures the multi-indices of $S_{i}$, $H_{i}$, $f_{s,i}$, and $f_{h,i}$ are the same; See Definition 14 in \cite{Watanabe_book_mtbs}.
\item[(C4)-(C5)]: Conditions (C4) and (C5) are mathematical conditions to control residuals in the proof. 
\item[(C6)-(C7)]: 
Conditions (C6) and (C7) ensures the well behaviours of empirical processes $\xi_{s,n}$ and $\xi_{h,n}$.
For the weighted inference,
these conditions are proved by using the Lindeberg--Feller-type central limit theorems. 
\item[(C8)]: Condition (C8) is a mathematical condition to control residuals in the proof.
\end{enumerate}

\section{Proof of Theorem \ref{thm: consistency of PCIC in regular models}}
\label{proofoftheorem}

In this appendix, we provide a sketch of the proof of the main theorem.
Here, we will derive the difference between 
the the weighted quasi-Bayesian generalisation error $G_{n}$ and
the weighted quasi-Bayesian training error
\begin{align*}
T_{n}:=
-\sum_{i=1}^{n}
    \frac{w_i}{n}\log \Eppos[h_{i}(X_{i}\,\mid\, \theta)].
\end{align*}
The proof consists of three steps.
First, we derive asymptotic forms of quasi posterior expectations by using the inverse Mellin transform.
We next expand weighted quasi-Bayesian generalisation and training errors by using the functional cumulant expansion.
Finally, we evaluate the difference by the Stein(--Malliavin) identity in the presence of correlation.

\textbf{The first step: the asymptotic quasi-posterior expectation}.
First, 
we shall derive an asymptotic form of the quasi-posterior distribution
by using the inverse Mellin transform (Chapter 4 of \cite{Watanabe_book}).
Let
\begin{align*}
\tilde{\Omega}(\theta):=\pi(\theta)\exp\left\{\sum_{i=1}^{n}
(s_{i}(Y_{i},\theta)-s_{i}(Y_{i},\theta_{0}))\right\}
\end{align*}
and we then express the quasi-posterior distribution as 
\begin{align*}
\pi(\theta\,;\,Y)=\frac{\tilde{\Omega}(\theta)}{\int \tilde{\Omega}(\theta') d\theta'}.
\end{align*}

From equation (\ref{eq: forms of f}),
we have, in each local coordinate of $\mathcal{M}$,
\begin{align*}
\tilde{\Omega}(\theta)d\theta
&=\tilde{\Omega}(g(u))|g'(u)|du\\
&=\exp\left\{-n u^{k}\frac{1}{n}\sum_{i=1}^{n} a_{s,i}(Y_{i},u)\right\}b(u)|u^{h}|du
\\
&=\exp\left\{
-nu^{k}\bar{A}_{s}(u)\left(
u^{k}
-\frac{1}{\sqrt{n}}\xi_{s}(u)
\right)
\right\}
b(u)|u^{h}|du
.
\end{align*}
We define the real log canonical threshold $\lambda$ and its multiplicity $m$ as follows:
\begin{align*}
\lambda&:=\min_{\text{all local coordinates}}
\min_{1\le j \le d}\frac{h_{j}+1}{2k_{j}}
\,\,\text{and}\\
m&=\max_{\text{all local coordinates}}
\#\left\{j:\, \frac{h_{j}+1}{2k_{j}}=\lambda\right\}.
\end{align*}
Then we rearrange the order of the indices so that for $u=(u_{a},u_{b})$,
$u_{a}$ is the part of $u$ of which the index satisfies
$\{h_{j}+1\}/\{2k_{j}\}=\lambda,\,j\in a$.

Letting $\delta(\cdot)$ be the Dirac delta function,
Theorem 9 of \cite{Watanabe_book_mtbs}
implies that, for $u$ in each local coordinate of $\mathcal{M}$,
\begin{align*}
\delta(t-u^{2k})|u^{h}|b(u)du=
t^{\lambda-1}(-\log t)^{m-1}du^{*} + o(t^{\lambda-1}(-\log t)^{m-1})
\end{align*}
where  $du^{*}$ is the integration
\begin{align*}
du^{*}=\frac{b(u)}{2^{m}(m-1)!\prod_{j=1}^{m}k_{j}}\delta(u_{a})(u_{b})^{\mu}du
\end{align*}
with $\mu:=\{\mu_{j}=-2\lambda k_{j}+h_{j}: \, m+1\le j \le d\}$.
Together with the partition of unity $\{\sigma_{\alpha}(\cdot):\alpha \text{ is an index of local coordinate}\}$
and letting $U_{\alpha}$ be the $\alpha$-th local coordinate of $\mathcal{M}$,
this gives, for an integrable function $F(\theta)$,
\begin{align*}
&\int F(\theta)\tilde{\Omega}(\theta)d\theta\\
&=\sum_{\alpha:\text{index of local coordinate}}
\int_{U_{\alpha}}\sigma_{\alpha}(u) F(g(u)) \tilde{\Omega}(g(u))|g'(u)|du
\\
&=\sum_{\alpha}
\int_{U_{\alpha}}\int_{0}^{\infty} 
\,F(t,u)\exp(
-\bar{A}_{s}(u)t+\bar{A}_{s}(u)\sqrt{t}\xi_{s,n}(u)
)
\,\delta(t-n u^{2k})
b_{\alpha}(u)|u^{h}|dtdu
\\
&=
\frac{(\log n)^{m-1}}{n^{\lambda}}
\sum_{\alpha}\int_{U_{\alpha}}
\int_{0}^{\infty}\, 
F(t,u)
t^{\lambda-1}\exp(
-\bar{A}_{s}(u)t+\bar{A}_{s}(u)\sqrt{t}\xi_{s,n}(u)
)dtdu^{**}
\\
&\quad\quad +
o_{P}\left(\frac{(\log n)^{m-1}}{n^{\lambda}}\right),
\end{align*}
where 
$F(t,u)$ is $F(\theta)$ regarded as the function with respect to $(t,u)$,
$b_{\alpha}(u):=\sigma_{\alpha}(u)b(u)$,
and $du^{**}=\sigma_{\alpha}(u)du^{*}$.
Thus, the quasi-posterior expectation of an integrable function $F(\theta)$ is expanded as 
\begin{align}
\Eppos[F(\theta)]= \langle  F(t,u) \rangle (1+o_{P}(1)) ,
\label{eq: asymptotic posterior expectation}
\end{align}
where
\begin{align*}
\langle
F(t,u)
\rangle
=
\frac{\sum_{\alpha} \int_{U_{\alpha}}\int_{0}^{\infty} \,F(t,u)t^{\lambda-1}\exp\left(
-\bar{A}_{s}(u)t+\bar{A}_{s}(u)\sqrt{t}\xi_{s,n}(u)
\right) dtdu^{**}}{
\sum_{\alpha} \int_{U_{\alpha}}\int_{0}^{\infty}\, t^{\lambda-1}\exp\left(
-\bar{A}_{s}(u)t+\bar{A}_{s}(u)\sqrt{t}\xi_{s,n}(u)
\right) dtdu^{**}}.
\end{align*}

\textbf{The second step: the functional cumulant expansion}.
Let us introduce the functional cumulant generating functions
\begin{align*}
\mathcal{G}(\beta)&=\Ep_{\tilde{Y}}\left[\frac{1}{n}\sum_{i=1}^{n}w_{i}\log \Eppos[h_{i}(\tilde{Y_{i}}\mid \theta)^{\beta}]\right]\quad\text{and}
\\
\mathcal{T}(\beta)
&=\frac{1}{n}\sum_{i=1}^{n} w_{i}\log \Eppos
\left[h_{i}(Y_{i}\mid\theta)^{\beta}\right].
\end{align*}
Then, the Taylor expansion gives us the following relations:
\begin{align*}
G_{n}&=-\mathcal{G}'(0)-\frac{1}{2}\mathcal{G}''(0)+O_{P}\left(\sup_{|\beta|\le 1} \left|\mathcal{G}'''(\beta)\right|\right)
\,\,\text{and}
\\
T_{n}&=-\mathcal{T}'(0)-\frac{1}{2}\mathcal{T}''(0)+O_{P}\left(\sup_{|\beta|\le 1}\left|\mathcal{T}'''(\beta)\right|\right).
\end{align*}
Consider the remaining terms. 
Let $\underline{w}:=\inf_{n}\min_{i}w_{i}$ and let $\overline{w}:=\sup_{n}\max_{i}w_{i}$.
From Lemma 8 of \cite{Watanabe_book_mtbs}, we get
\begin{align*}
\sup_{|\beta|\le 1} \left|\mathcal{G}'''(\beta)\right|
&\le
\frac{6}{\underline{w}^{2}}\frac{1}{n}\sum_{i=1}^{n}
\Ep_{\tilde{Y}}\left[
\frac{\Eppos[|f_{h,i}(\tilde{Y}_{i},\theta)|^{3}\exp(-\beta f_{h,i}(\tilde{Y}_{i},\theta)) / \overline{w}]}
{\Eppos[\exp(-\beta f_{h,i}(\tilde{Y}_{i},\theta) / \underline{w})]}
\right].
\end{align*}
From equation (\ref{eq: forms of f}),
we further get
\begin{align*}
\sup_{|\beta|\le 1} \left|\mathcal{G}'''(\beta)\right|
&\le
\frac{6}{\underline{w}^{2}}\frac{1}{n}\sum_{i=1}^{n}
\Ep_{\tilde{Y}}\left[
\sup_{u\in\mathcal{M}}|a_{h,i}(\tilde{Y}_{i},u)|^{3}
\frac{\Eppos[|u^{k}|^{3}\exp(-\beta f_{h,i}(\tilde{Y}_{i},\theta)) / w_{i}]}
{\Eppos[\exp(-\beta f_{h,i}(\tilde{Y}_{i},\theta) / w_{i})]}
\right],
\end{align*}
and here,
by considering Conditions (C4) and (C8),
and
by applying Lemmas 17-19 of \cite{Watanabe_book_mtbs} 
we obtain
\begin{align*}
\frac{1}{n}\sum_{i=1}^{n}
\Ep_{\tilde{Y}}\left[
\sup_{u\in\mathcal{M}}|a_{h,i}(\tilde{Y}_{i},u)|^{3}
\frac{\Eppos[|u^{k}|^{3}\exp(-\beta f_{h,i}(\tilde{Y}_{i},\theta)) / w_{i}]}
{\Eppos[\exp(-\beta f_{h,i}(\tilde{Y}_{i},\theta) / w_{i})]}
\right]
=O_{P}(n^{-3/2}),
\end{align*}
which implies 
$\sup_{|\beta|\le 1} \left|\mathcal{G}'''(\beta)\right|
=O_{P}(n^{-3/2})$ and $\sup_{|\beta|\le 1} \left|\mathcal{T}'''(\beta)\right|
=O_{P}(n^{-3/2})$.

\textbf{The third step: evaluating $\mathcal{G}'(0)$, $\mathcal{G}''(0)$, $\mathcal{T}'(0)$, and $\mathcal{T}''(0)$}. 
By definition, we have
\begin{align*}
-\mathcal{G}'(0)
&=
\Ep_{\tilde{Y}}\left[
-\frac{1}{n}\sum_{i=1}^{n}w_{i}\log h_{i}(\tilde{Y}_{i}\mid\theta_{0})\right]
+
\Ep_{\tilde{Y}}
\Eppos\left[
\frac{1}{n}\sum_{i=1}^{n}f_{h,i}(\tilde{Y}_{i}, \theta)\right],
\\
-\mathcal{T}'(0)
&=
-\frac{1}{n}\sum_{i=1}^{n}
w_{i} \log h_{i}(Y_{i}\mid\theta_{0})
+\Eppos\left[\frac{1}{n}\sum_{i=1}^{n}f_{h,i}(Y_{i},\theta)\right]
,
\\
-\mathcal{G}''(0)
&=-\Ep_{\tilde{Y}}
\left[
\frac{1}{n}\sum_{i=1}^{n}w_{i}
\Vpos[\log h_{i}(\tilde{Y}_{i}\mid\theta)]
\right],\,\,\text{and}
\\
-\mathcal{T}''(0)
&=
-\frac{1}{n}\sum_{i=1}^{n}w_{i}
\Vpos[\log h_{i}(Y_{i}\mid \theta)].
\end{align*}
It is easy to see that 
$\Ep_{Y}[\mathcal{G}''(0)-\mathcal{T}''(0)]=o(n^{-1})$ and so we shall analyse 
the difference $\Ep_{Y}[-\mathcal{T}'(0)+\mathcal{G}'(0)]$.
Equation (\ref{eq: forms of f}) gives
\begin{align*}
&\Eppos\left[\frac{1}{n}\sum_{i=1}^{n}f_{h,i}(Y_{i},\theta)\right]\\
&=
\Eppos\left[ \frac{1}{n} \sum_{i=1}^{n} a_{h,i}(Y_{i},u) 
u^{k}\right]
\\
&=
\Eppos\left[ \bar{A}_{h}(u)u^{2k} \right]
-
\frac{1}{n}
\left\langle 
\bar{A}_{h}(u) \sqrt{t} \xi_{h,n}(u)\right\rangle 
+o_{P}\left(
n^{-1}
\right)\\
&=\Ep_{\tilde{Y}}
\Eppos\left[
\frac{1}{n}\sum_{i=1}^{n}f_{h,i}(\tilde{Y}_{i}, \theta)\right]
-\frac{1}{n}
\left\langle 
\bar{A}_{h}(u) \sqrt{t} \xi_{h,n}(u)\right\rangle 
+o_{P}\left(
n^{-1}
\right).
\end{align*}
For $\lambda>0$, and for an integrable function $F(t,u)$, let 
\begin{align*}
\mathcal{S}_{\lambda}[F(t,u)]
:=  \sum_{\alpha}\int_{U_{\alpha}} \int_{0}^{\infty} t^{\lambda-1}\exp(-\bar{A}_{s}(u) t) F(t,u) dt du^{**}.
\end{align*}
Together with Condition (C7), this definition gives
\begin{align*}
\Ep_{Y}\left[\left\langle 
\bar{A}_{h}(u) \sqrt{t} \xi_{h,n}(u)\right\rangle \right]
=\Ep_{\xi_{h},\xi_{s}}\left[\frac{\mathcal{S}_{\lambda}[\bar{A}_{h}(u)\sqrt{t}\xi_{h}\exp(\bar{A}_{s}(u)\sqrt{t}\xi_{s})]}
{\mathcal{S}_{\lambda}[\exp(\bar{A}_{s}(u')\sqrt{t'}\xi_{s}(u'))]}\right]
+\Ep_{Y}[o_{P}(1)],
\end{align*}
where $\Ep_{\xi_{h},\xi_{s}}$ is the expectation with respect to $\xi_{h}$ and $\xi_{s}$.
To analyse the above,
we employ the Stein(--Malliavin) identity in the presence of correlation:
Using the orthogonal normal basis $\{\phi_{k}\}$ of $L^{2}(\mathcal{M})$,
we conduct the Karhunen-Lo\`{e}ve expansions of $\xi_{s}$ and $\xi_{h}$:
\begin{align*}
\xi_{s}(u)
=\sum_{k=1}^{\infty}\phi_{k}(u)\sqrt{\kappa^{s}_{k}}\xi_{s,k}
\quad\text{and}\quad
\xi_{h}(u)=
\sum_{k=1}^{\infty}\phi_{k}(u)\sqrt{\kappa^{h}_{k}}\xi_{h,k},
\end{align*}
where 
\begin{align*}
\kappa^{s}_{k}:= \mathbb{V}_{\xi_{s}}\left[\int \phi_{k}(u)\xi_{s}(u)du\right]
\quad\text{and}\quad
\kappa^{h}_{k}:= \mathbb{V}_{\xi_{h}}\left[\int \phi_{k}(u)\xi_{h}(u)du\right]
\end{align*}
with $\mathbb{V}_{\xi_{h}}$
and $\mathbb{V}_{\xi_{s}}$
denoting the variances with respect to $\xi_{h}$ and $\xi_{s}$, respectively.
Here $\{\xi_{s,k}\}$
are independent Gaussian with the variance 1,
and $\{\xi_{h,k}\}$
are independent Gaussian with the variance 1.
But,
$\xi_{s,k}$
and $\xi_{h,k'}$ are correlated, and we denote their correlations by $\kappa_{kk'}$:
$\kappa_{kk'}=\Ep_{\xi_{s},\xi_{h}}[\xi_{s,k}\xi_{h,k}]$.
Then, 
by letting $\phi^{h,*}_{k}(u):=\sqrt{\kappa_{k}^{h}}\phi_{k}(u)$ and by the Karhunen-Lo\'{e}ve expansion of $\xi_{h}$,
we have
\begin{align*}
&\Ep_{\xi_{h},\xi_{s}}
\left[
\frac{\mathcal{S}_{\lambda}[\bar{A}_{h}(u)\sqrt{t}\xi_{h}(u)\exp(\bar{A}_{s}(u)\sqrt{t}\xi_{s}(u) )]}
{\mathcal{S}_{\lambda}[\exp(\bar{A}_{s}(u')\sqrt{t'}\xi_{s}(u'))]}
\right]\\
&=
\sum_{k'=1}^{\infty}
\mathcal{S}_{\lambda}\left[
\Ep_{\xi_{h},\xi_{s}}\left[
\xi_{h,k'}
\frac{
\bar{A}_{h}(u)\sqrt{t}\phi_{k'}^{*}(u)\exp(\bar{A}_{s}(u)\sqrt{t}\xi_{s}(u) )]}
{\mathcal{S}_{\lambda}[\exp(\bar{A}_{s}(u')\sqrt{t'}\xi_{s}(u'))]}
\right]
\right].
\end{align*}
Here we use the Stein identity in the presence of correlation:
Suppose that $(N_{1},N_{2})$ follows a bivariate Gaussian distribution and $F$ is a differentiable function with $\Ep[|F'(N_{1})|]$. Then, 
$\Cov[F(N_{1}),N_{2}]=\Cov[N_{1},N_{2}]\Ep[F'(N_{1})]$.
Together with 
letting $\kappa_{kk'}$ be the correlation between $\xi_{s,k}$ and $\xi_{h,k'}$
and letting
$\phi^{s,*}_{k}(u):=\sqrt{\kappa_{k}^{s}}\phi_{k}(u)$,
this identity gives
\begin{align*}
&\Ep_{\xi_{h},\xi_{s}}\left[
\xi_{h,k'}
\frac{
\bar{A}_{h}(u)\sqrt{t}\phi_{k'}^{h,*}(u)\exp(\bar{A}_{s}(u)\sqrt{t}\xi_{s}(u) )]}
{\mathcal{S}_{\lambda}[\exp(\bar{A}_{s}(u')\sqrt{t'}\xi_{s}(u'))]}
\right]
\\
&=
\Ep_{\xi_{h},\xi_{s}}
\left[
\sum_{k=1}^{\infty}\kappa_{kk'}
\frac{\partial}{\partial \xi_{s,k}}
\frac{
\bar{A}_{h}(u)\sqrt{t}\phi_{k'}^{h,*}(u)\exp(\bar{A}_{s}(u)\sqrt{t}\xi_{s}(u) )]}
{\mathcal{S}_{\lambda}[\exp(\bar{A}_{s}(u')\sqrt{t'}\xi_{s}(u'))]}
\right]
\\
&=
\Ep_{\xi_{h},\xi_{s}}\Bigg[
\sum_{k=1}^{\infty}
\kappa_{kk'}
\Bigg\{
\frac{
\bar{A}_{h}
(u)
\bar{A}_{s}(u)
t\phi_{k'}^{h,*}(u)
\phi_{k}^{s,*}(u)
\exp(\bar{A}_{s}(u)\sqrt{t}\xi_{s}(u) )
}{\mathcal{S}_{\lambda}[\exp(\bar{A}_{s}(u')\sqrt{t'}\xi_{s}(u'))]}
\\
&\quad\quad\quad\quad
-
\frac{
\bar{A}_{h}(u)\sqrt{t}\phi_{k'}^{h,*}(u)
\mathcal{S}_{\lambda}[\bar{A}_{s}(u')\sqrt{t'}\phi^{s,*}_{k'}(u')\exp(\bar{A}_{s}(u')\sqrt{t'}\xi_{s}(u'))]
}
{(\mathcal{S}_{\lambda}[\exp(\bar{A}_{s}(u')\sqrt{t'}\xi_{s}(u'))])^{2}}
\Bigg\}
\Bigg].
\end{align*}
Thus,
from (\ref{eq: asymptotic posterior expectation}) and Condition (C7),
we have
\begin{align*}
&\Ep_{\xi_{h},\xi_{s}}
\left[
\frac{\mathcal{S}_{\lambda}[\bar{A}_{h}(u)\sqrt{t}\xi_{h}(u)\exp(\bar{A}_{s}(u)\sqrt{t}\xi_{s}(u) )]}
{\mathcal{S}_{\lambda}[\exp(\bar{A}_{s}(u')\sqrt{t'}\xi_{s}(u'))]}
\right]\\
&=
\Ep_{\xi_{s}}
\sum_{k,k'=1}^{\infty}
\kappa_{kk'}
\langle\bar{A}_{h}(u)\sqrt{t}\phi^{h,*}_{k'}(u),\bar{A}_{s}(u)\sqrt{t}\phi^{s,*}_{k}(u)\rangle_{\infty} ,
\end{align*}
where $\langle\cdot, \cdot\rangle_{\infty}$ is the asymptotic quasi-posterior covariance defined as
\begin{align*}
\langle F(t,u),G(t,u)\rangle_{\infty}
:=
\langle F(t,u)G(t,u)\rangle_{\infty}
-\langle F(t,u)\rangle_{\infty}
\langle G(t,u) \rangle_{\infty}
\end{align*}
with $\langle \cdot \rangle_{\infty}$ defined by replacing $\xi_{s,n}$ with $\xi_{s}$ in $\langle \cdot \rangle$. 
Consider the first term on the right hand side of the above.
From the definition of the Karhunen-Lo\`{e}ve expansion,
we have
\begin{align*}
\sum_{k,k'=1}^{\infty}
\lambda_{kk'}
\phi^{h,*}_{k'}(u)
\phi^{s,*}_{k}(u)
=
\Ep_{\xi_{h},\xi_{s}}
[\xi_{h}(u)\xi_{s}(u)].
\end{align*}
Since $u^{k}=0$ on the support of $du^{**}$,
we have, for $u$ contained in the support of $du^{**}$, 
\begin{align*}
\sum_{k,k'=1}^{\infty}
\kappa_{kk'}
\phi^{h,*}_{k'}(u)
\phi^{s,*}_{k}(u)
&=\frac{1}{n}\Ep_{\tilde{Y}}\left[\sum_{i=1}^{n}\tilde{a}_{h,i}(\tilde{Y}_{i},u)
\sum_{i=1}^{n}\tilde{a}_{s,i}(\tilde{Y}_{i},u)
\right]
\\
&=
\frac{1}{n}\Ep_{\tilde{Y}}\left[\sum_{i=1}^{n}\tilde{a}_{h,i}(\tilde{Y}_{i},u)
\tilde{a}_{s,i}(\tilde{Y}_{i},u)
\right]
\\
&=
\frac{1}{n}
\left(\sum_{i=1}^{n}\tilde
{a}_{h,i}(Y_{i},u)
\tilde{a}_{s,i}(Y_{i},u)
\right)
+o_{P}(1),
\end{align*}
where the last equation follows from Condition (C6).
Further, we have
\begin{align*}
&\sum_{k,k'=1}^{\infty}
\kappa_{kk'}
\langle \bar{A}_{h}(u)\sqrt{t}\phi^{h,*}_{k'}(u)\rangle_{\infty}
\langle \bar{A}_{s}(u)\sqrt{t}\phi^{s,*}_{k}(u)\rangle_{\infty}
\\
&=
\Ep_{\tilde{\xi}_{h},\tilde{\xi}_{s}}
\left[\sum_{k'=1}^{\infty}
\tilde{\xi}_{h,k'}\langle \bar{A}_{h}(u)\sqrt{t}\phi_{k'}^{h,*}(u)
\rangle_{\infty}
\sum_{k=1}^{\infty}
\tilde{\xi}_{s,k}\langle
\bar{A}_{s}(u)\sqrt{t}
\phi_{k}^{s,*}(u)
\rangle_{\infty}\right]
\\
&=
\frac{1}{n}
\Ep_{\tilde{Y}}
\left[
\left\langle \sum_{i=1}^{n}a_{h,i}(\tilde{Y}_{i},u)u^{k}
\right\rangle
\left\langle \sum_{i=1}^{n}a_{s,i}(\tilde{Y}_{i},u)u^{k}
\right\rangle
\right]
+o(1)
\\
&=
\Ep_{\tilde{Y}}\left[
\frac{1}{n}
\sum_{i=1}^{n}
\langle
\langle a_{h,i}(\tilde{Y}_{i},u) u^{k}
 a_{s,i}(\tilde{Y}_{i},v) v^{k}
\rangle
\rangle
\right]
\\
&=\frac{1}{n}
\sum_{i=1}^{n}
\langle\langle a_{h,i}(Y_{i},u) u^{k} a_{s,i}(Y_{i},v) v^{k}
\rangle\rangle
+o_{P}(1),
\end{align*}
where
$\langle\langle F(u,v)\rangle\rangle:=
\langle\langle F(u,v)\rangle_{\text{with respect to $u$}}\rangle_{\text{with respect to $v$}}$
for a function $F(u,v)$,
and the last equation follows from Condition (C6).
Thus, we obtain
\begin{align*}
\Ep_{Y}\left[\left\langle 
\bar{A}_{h}(u) \sqrt{t} \xi_{h,n}(u)\right\rangle \right]
&=\frac{1}{n}
\sum_{i=1}^{n}
\Covpos[ a_{h,i}(Y_{i},u)u^{k}, a_{s,i}(Y_{i},u)u^{k}]
+o_{P}(1)\\
&=\frac{1}{n}\sum_{i=1}^{n}\Covpos[f_{h,i}(Y_{i},\theta),f_{s,i}(Y_{i},\theta)] + o_{P}(1)\\
&=\frac{1}{n}
\sum_{i=1}^{n}
w_{i}\Covpos[\log h_{i}(Y_{i}\mid \theta), s_{i}(Y_{i},\theta)]
+ o_{P}(1),
\end{align*}
which completes the proof.
\qed
\end{appendices}

\end{document}